\documentclass[11pt]{article}
\usepackage[titletoc]{appendix}
\usepackage[labelsep=period]{caption}
\usepackage{bm}
\usepackage{fullpage}
\usepackage{epsfig}
\usepackage{graphicx}
\usepackage{subcaption}
\usepackage{latexsym}
\usepackage{amsmath}
\usepackage{amsfonts}
\usepackage{amssymb}
\usepackage{mathrsfs}
\usepackage{pifont}
\usepackage{yhmath}
\usepackage{bbm}
\usepackage{dsfont}
\usepackage{eufrak}
\usepackage{wasysym}
\usepackage{underscore}
\usepackage{epstopdf}
\usepackage{fontenc}
\usepackage{amsthm}
\newtheorem{theorem}{Theorem}[section]

\newtheorem{lemma}[theorem]{Lemma}

\DeclareMathAlphabet{\mathpzc}{OT1}{pzc}{m}{it}

\def\qed{\hfill \rule{4pt}{7pt}}
\usepackage{fullpage,graphicx}

\DeclareMathAlphabet{\mathpzc}{OT1}{pzc}{m}{it}
\title{\bf Proof of the Langberg--M\'edard Multiple Unicast Conjecture}
\author{\vspace{2mm}  Kai Cai$^{a}$
\quad Guangyue Han$^{a}$\thanks{Supported in part by a grant from the Research Grants Council of Hong Kong (Project No. 17303825).}
\quad Mengxi Yang$^{b}$\thanks{Corresponding author. E-mail: yangmx221b@outlook.com.}
\quad Wenan Zang$^{a}$\thanks{Supported in part by a grant from the Research Grants Council of Hong Kong (Project No. 17304525).} \\
$\stackrel{a}{}$ Department of Mathematics \\ University of Hong Kong \\ Hong Kong, China \smallskip\\
$\stackrel{b}{}$ School of Mathematical Sciences \\ University of Science and Technology of China \\ Hefei 230026, China}
\begin{document}
\date{}
\maketitle

\begin{abstract}
Let $\mathcal{N}=(V,A,S, R)$ be a directed network, where $D=(V,A)$ is a digraph, $S=\{s_1,s_2,\ldots,s_k\}\subseteq V$ is a set
of senders, and $R=\{r_1,r_2,\ldots,r_k\}\subseteq V$ is a set of receivers. We prove that if, for each $1\le j \le k$, there 
exist $k$ pairwise arc-disjoint directed paths in $D$, one from each $s_i$ to $r_j$, then the underlying undirected network of ${\cal N}$, 
with unit capacities, admits a feasible $(s_1,s_2,\ldots,s_k)$--$(r_1,r_2,\ldots,r_k)$ multiflow of value $1$ for each pair $(s_i,r_i)$, 
$i=1,2,\ldots,k$. This establishes the multiple unicast conjecture proposed by Langberg and M\'edard in 2009.

\vskip 4mm



\end{abstract}

\newpage
\section{Introduction}

Network coding is a communication paradigm in which intermediate nodes may perform coding operations on the data they receive, combining 
multiple information flows before forwarding them in order to improve throughput, reliability, or efficiency. The multiple unicast problem 
asks whether independent messages can be transmitted simultaneously from each sender to its corresponding receiver over a shared network. A 
rate vector specifies the rates achieved simultaneously by the sender-receiver pairs. When transmission is restricted to routing along paths, 
it can be modeled by a feasible multiflow $f$, whose associated rate vector records the rates carried by $f$. A central question in network coding 
theory is whether, and to what extent, allowing intermediate nodes to code can enlarge the set of achievable rate vectors beyond those achievable 
by feasible multiflows. The {\em multiple unicast conjecture} of Li and Li~\cite{LiLi} asserts that, in undirected networks, network coding 
and feasible multiflow routing achieve exactly the same set of rate vectors. As observed by Adler et al. \cite{A}, the validity of this conjecture would 
also provide an affirmative answer to a long-standing open problem in I/O complexity and yield nontrivial lower bounds in the computational 
complexity of oblivious Turing machines. This conjecture resisted tremendous research efforts for more than two decades (see, for instance, 
\cite{A,BC,CH1,CH2,CH3,CH4,H,J,LM,LiLi,LW,QT,XL,Y,ZLL}) until recently, when Braverman and He \cite{BH} and Zhang, Li, and Li \cite{ZLL} 
independently constructed counterexamples to it.

The present paper is devoted to the study of a weaker version of this conjecture. Consider a directed network $\mathcal{N}=(V,A,S, R)$, where 
$D=(V,A)$ is a digraph, $S=\{s_1,s_2,\ldots, s_k\}\subseteq V$ is a set of senders, and $R=\{r_1,r_2,\ldots,r_k\}\subseteq V$ is a set of receivers,
with $S \cap R=\emptyset$. We say that $\mathcal{N}$ is {\em strongly reachable} if, for each $1\le j \le k$, there exist $k$ pairwise arc-disjoint 
directed paths in $D$, one from each $s_i$ to $r_j$. Let $G=(V,E)$ be the underlying undirected graph of $D$, which arises from $D$ by disregarding 
the directions of all arcs. We say that $\overline{\cal N}:=(V,E,S,R)$ is the underlying undirected network of ${\cal N}$. Throughout we assume 
that $G$ is connected and each edge of $\overline{\cal N}$ has unit capacity.

For each $1\le i \le k$, let ${\cal P}_i$ be the set of all simple $s_i$-$r_i$ paths (undirected) in $G$. An {\em $(s_1,s_2,\ldots,s_k)$--$(r_1,r_2, \ldots,r_k)$ 
multiflow} in $\overline{\cal N}$ is a vector-valued function $f=(f_1, f_2,\ldots,f_k)$ such that $f_i: {\cal P}_i\to{\mathbb R}_+$ is an $s_i$-$r_i$ 
flow for each $i$. (As usual, ${\mathbb R}_+$ stands for the set of all nonnegative real numbers.) We call $f$ {\em feasible} if $\sum_{i=1}^k \sum_{e\in P\in{\cal P}_i} 
f_i(P) \leq 1$ for each $e\in E$, and call the corresponding vector $(\operatorname{val}(f_1), \operatorname{val}(f_2),\ldots,\operatorname{val}(f_k))$ the 
{\em rate} of $f$, where $\operatorname{val}(f_i):=\sum_{P\in{\cal P}_i}f_i(P)$ is the value of $f_i$ for each $i$. 

We are now ready to present the aforementioned weaker version formulated by Langberg and M\'edard \cite{LM} in 2009.

\vskip 3mm
\noindent {\bf Multiple Unicast Conjecture} (Langberg and M\'edard \cite{LM}). {\em For every strongly reachable network ${\cal N}$, 
there is a feasible multiflow in $\overline{\cal N}$ with rate $(1,1,\ldots,1)$.}
\vskip 3mm

Langberg and M\'edard \cite{LM} showed that strong reachability of ${\cal N}$ guarantees the existence a feasible multiflow $f$ in 
$\overline{\cal N}$ with rate $(\frac{1}{3},\frac{1}{3}, \ldots, \frac{1}{3})$. Cai and Han improved this reachable rate to $(\frac{8}{9},\frac{8}{9},
\ldots,\frac{8}{9})$ in \cite{CH1}, further obtained $(\frac{11}{12},\frac{11}{12},\ldots,\frac{11}{12})$ for $k=3,4$ in \cite{CH2}, and 
confirmed the Langberg--M\'edard conjecture for $k=3$ in \cite{CH3}. The purpose of this paper is to prove the Langberg--M\'edard conjecture 
for every $k$.
\newpage

\begin{theorem}\label{thm:main}
Let ${\cal N}$ be a strongly reachable network with $k$ sender-receiver pairs $(s_i,r_i)$ for $i=1,2,\ldots,k$. Then there is a feasible 
$(s_1,s_2,\ldots,s_k)$--$(r_1,r_2,\ldots,r_k)$ multiflow in $\overline{\cal N}$ of value $1$ for each pair $(s_i,r_i)$.
\end{theorem}

We point out that the statement no longer holds if we replace strong reachability in our theorem by that $G$ contains $k$ edge-disjoint paths 
(undirected) to each receiver $r_j$, one from each sender $s_i$; a counterexample is given by $G:=s_1r_2r_1s_2$ (a path with four vertices).

Onaga \cite{O} obtained the following characterization for the existence of feasible multiflows in a general setting, which can be derived from the 
duality theorem of linear programming (see also Schrijver, Theorem 70.1 \cite{s03} and Onaga and Kakusho \cite{Ona}).

\begin{theorem}[Onaga \cite{O}]\label{thm:length}
Let $G=(V,E)$ be a supply graph with a nonnegative capacity $c(e)$ on each $e\in E$, and let $H=(T, R)$ be a demand graph with a nonnegative 
demand $d_i$ on each $s_it_i \in R$. Then there exists a feasible multiflow $f=(f_1, f_2, \ldots, f_k)$, with $k:=|R|$, subject to $c$ 
and satisfying $\operatorname{val}(f_i)=d_i$ for each $i$, if and only if for each length function $\ell: E\to {\mathbb R}_+$,
\[\sum_{i=1}^k d_i \cdot d_\ell(s_i,t_i) \leq\sum_{e\in E} c(e)\ell(e),\]
where $d_\ell(s_i,t_i)$ denotes the length of a shortest $s_i$-$t_i$ path with respect to $\ell$ in $G$.
\end{theorem}

When we restrict our attention to the undirected network $\overline{\cal N}$ specified above, the demand graph $H$ consists only of $k$ disjoint
edges $s_ir_i$  (which are not necessarily edges of $G$). Furthermore, the capacity $c(e)=1$ for each $e \in E$, and the demand $d_i=1$ 
for each $1\le i \le k$. In view of Theorem \ref{thm:length}, it suffices to establish the following statement in order to prove Theorem \ref{thm:main}.

\begin{theorem}\label{cor:length}
Let ${\cal N}$ be a strongly reachable network with $k$ sender-receiver pairs $(s_i,r_i)$ for $i=1,2,\ldots,k$. Then $\overline{\cal N}$ 
satisfies the following inequality for each length function $\ell: E\to {\mathbb R}_+$,
\[\sum_{i=1}^k d_\ell(s_i,r_i) \leq\sum_{e\in E} \ell(e).\]
\end{theorem}

For this length function $\ell$, let $d:=d_\ell$ throughout this paper. For the ends $v,w$ of an arbitrary edge $e$ of $G$, the triangle 
inequality gives $|d(v,r_j)-d(w,r_j)|\leq d(v,w)\leq\ell(e)$ for $1\le j\le k$. This motivates us to construct functions on $V$ based on the distances 
to the receivers, such that the sum of their absolute changes between any two vertices $v,w$ is at most $d(v,w)$. These functions will 
also allow us to bound the left-hand side of the preceding inequality by two terms, each of which will be estimated using a different 
inequality for cut size.

The remainder of this paper is organized as follows. In Section 2, we derive two cut inequalities from strong reachability and identify the 
properties required of the functions. In Section 3, we prescribe the values of the required functions at the receivers. 
In Section 4, we extend these values to all vertices and complete the proof of Theorem \ref{cor:length}. 

\section{Cut Bounds and Required Functions}

Let $\mathcal{N}=(V,A,S,R)$ and $\overline{\cal N}:=(V,E,S,R)$ be the networks specified in Theorem \ref{cor:length}. Recall that $G=(V,E)$ is 
the underlying undirected graph of $D=(V,A)$. For each $X\subseteq V$, write $\bar{X}:=V-X$. Let $\delta_D^+(X)$ be the set of all arcs of 
$D$ directed from $X$ to $\bar{X}$, and let $\delta_G(X)$ be the set of all edges of $G$ between $X$ and $\bar{X}$. We call $\delta_D^+(X)$
the {\em outgoing directed cut} defined by $X$ in $D$, and call $\delta_G(X)$ the undirected cut defined by $X$ in $G$. Set $s(X):=|X \cap S|$, 
which is the number of senders contained in $X$. 

The following lemma gives lower bounds on the sizes of the above cuts.

\begin{lemma}\label{lem:cuts}
For each $X\subseteq V$, the following two statements hold.
\begin{itemize}
\vspace{-1mm}
\item[(i)] If a receiver lies outside $X$, then $|\delta_D^+(X)|\geq s(X)$.
\vspace{-1mm}
\item[(ii)] If both $X$ and $\bar{X}$ contain a receiver, then $|\delta_G(X)|\geq k$.
\end{itemize}
\end{lemma}

{\bf Proof}. (i) By hypothesis, $r_j \notin X$ for some $1 \le j \le k$. By strong reachability of $\mathcal{N}$, there exist $k$ pairwise 
arc-disjoint directed paths in $D$, one from each $s_i$ to $r_j$. For each $s_i\in X$, the corresponding path begins in $X$ and ends in $\bar{X}$, 
so it contains at least one arc $e_i$ in $\delta_D^+(X)$. These arcs $e_i$, for $s_i\in X$, are distinct, as the paths we consider are 
arc-disjoint. It follows that $|\delta_D^+(X)|\geq s(X)$. 

\vskip 1mm
(ii) By hypothesis and (i), we have $|\delta_D^+(X)|\geq s(X)$ and $|\delta_D^+(\bar{X})|\geq s(\bar{X})$. So $|\delta_G(X)|=|\delta_D^+(X)|+
|\delta_D^+(\bar{X})|\geq s(X)+s(\bar{X})=k$, as desired. \qed\\

Given a real number $x$, its {\em positive part} is defined by $[x]_+:=\max\{x,0\}$. For a real-valued function $f$ on $V$, the set 
$\{v\in V:f(v)\geq t\}$ is called the {\em upper level set} of $f$ at level $t$. A nonnegative function that vanishes at every receiver has 
no receiver in any positive level set, so Lemma \ref{lem:cuts}(i) applies. A function taking both its minimum and maximum at receivers has receivers 
on both sides of each intermediate level set, so Lemma \ref{lem:cuts}(ii) applies. The following lemma makes both observations quantitative.

\begin{lemma}\label{lem:levels}
The following two statements hold.
\begin{itemize}
\vspace{-1mm}
\item [(i)] If a function $q: V\to{\mathbb R}_+$ is $0$ at every receiver, then 
\vspace{-3mm}
\[\sum_{i=1}^k q(s_i)\leq \sum_{(v,w)\in A} \, [q(v)-q(w)]_+.\]
\vspace{-5mm}
\item [(ii)] If a function $g: V\to[0,\lambda]$, where $\lambda\geq0$, takes the value
$\lambda$ at some receiver and the value $0$ at another receiver when $\lambda> 0$, then

\vspace{-3mm}
\[ k\lambda\leq\sum_{(v,w)\in A} \, |g(v)-g(w)|.\]
\vspace{-2mm}
\end{itemize}
\end{lemma}

\vspace{-2mm}

{\bf Proof}. (i) We may assume that $q$ is not identically zero, otherwise the inequality holds trivially. Since $V$ is finite and $q$ vanishes at every receiver, 
its distinct values can be listed as $0=t_0<t_1<\cdots<t_N$, for some $N\geq 1$. For each $1\le j\le N$, let $X_j:=\{v\in V:q(v)\geq t_j\}$ and $\alpha_j:=t_j-t_{j-1}$.
Since $t_j>0$, every receiver lies outside $X_j$. By Lemma \ref{lem:cuts}(i), we have $|\delta_D^+(X_j)|\geq s(X_j)$. Multiplying both sides of this inequality by 
$\alpha_j$, and summing over all $j$ gives

\vskip 1mm
(1) $ \sum_{j=1}^N\alpha_j s(X_j)  \leq \sum_{j=1}^N\alpha_j|\delta_D^+(X_j)|.$

\vskip 1mm
Suppose $q(s_i)=t_m$. For $m \ge 1$, the sender $s_i$ is contained only in the sets $X_1, X_2,\ldots, X_m$, so its total contribution to
$\sum_{j=1}^N\alpha_j s(X_j)$ is $\sum_{j=1}^m\alpha_j  =(t_1-t_0)+\cdots+(t_m-t_{m-1})  = t_m=q(s_i)$.
For $m=0$, this is a sum with no terms, which has value $0$. Thus

\vskip 1mm
(2) $ \sum_{j=1}^N\alpha_j s(X_j)  =\sum_{i=1}^k\ \sum_{1\leq j\leq N, \, s_i\in X_j} \alpha_j  =\sum_{i=1}^kq(s_i).$

\vskip 1mm
Consider an arbitrary arc $(v,w)\in A$. It belongs to $\delta^+_D(X_j)$ if and only if $q(w)<t_j\leq q(v)$. Let $t_n:=q(v)$ and $t_m:=q(w)$.
If $t_n>t_m$, then this arc is counted only in $|\delta_D^+(X_j)|$ for $j=m+1, m+2, \ldots, n$, so its total contribution to $\sum_{j=1}^N\alpha_j|\delta_D^+(X_j)|$ is 
$\sum_{j=m+1}^n\alpha_j=t_n-t_m=q(v)-q(w)$. If $t_n \leq t_m$, this arc leaves none of the sets $X_j$ and hence contributes $0$ to 
$\sum_{j=1}^N\alpha_j|\delta_D^+(X_j)|$. Therefore

\vskip 1mm
(3) $\sum_{j=1}^N\alpha_j|\delta_D^+(X_j)| =\sum_{(v,w)\in A} \ \sum_{1\leq j\leq N, \, q(w)<t_j\leq q(v)}\alpha_j
 =\sum_{(v,w)\in A} \, [q(v)-q(w)]_+.$
\vskip 2mm

Substituting (2) and (3) into (1) establishes the desired inequality.

\vskip 1mm

(ii) We may assume that $\lambda>0$, otherwise the desired inequality holds trivially. 
Let us list the distinct values of $g$ as $0=t_0<t_1<\cdots<t_N=\lambda$. Define $X_j:=\{v\in V: g(v)\geq t_j\}$ and 
$\alpha_j:=t_j-t_{j-1}$ for $1\le j \le N$. Then the receiver with value $\lambda$ belongs to every $X_j$, and the 
receiver with value $0$ belongs to every $\bar{X}_j$. By Lemma \ref{lem:cuts}(ii), we have $|\delta_G(X_j)|\geq k$ for 
each $1\le j \le N$. Multiplying both sides of this inequality by $\alpha_j$ and summing over all $j$, we obtain
\vskip 1mm

(4) $ \sum_{j=1}^N\alpha_j|\delta_G(X_j)| \geq k\sum_{j=1}^N\alpha_j  =k(t_N-t_0)=k\lambda$.

\vskip 1mm

Consider an arbitrary edge $vw \in E$. It belongs to $\delta_G(X_j)$ if and only if $\min\{g(v),g(w)\}<t_j\leq\max\{g(v),g(w)\}$.
Let $t_m:=\min\{g(v),g(w)\}$ and $t_n:=\max\{g(v),g(w)\}$. If $t_n>t_m$, the contribution of this edge to $\sum_{j=1}^N\alpha_j|\delta_G(X_j)|$ is
$\alpha_{m+1}+\alpha_{m+2}+\cdots+\alpha_n=t_n-t_m=|g(v)-g(w)|$. If $t_n=t_m$, the sum has no terms and is $0$. Adding these contributions 
over all edges yields

\vskip 1mm
(5) $\sum_{j=1}^N\alpha_j |\delta_G(X_j)|  =\sum_{vw \in E} |g(v)-g(w)| =\sum_{(v,w)\in A} |g(v)-g(w)|.$

\vskip 1mm
Substituting (5) into (4) establishes the desired inequality. \qed\\

To derive the inequality in Theorem \ref{cor:length}, we seek functions $q,u_1,u_2,u_3$ to which both inequalities in Lemma~\ref{lem:levels} are 
applicable. The sum of their absolute changes between any two vertices must be bounded by the distance between these two vertices, so that 
adding the cut inequalities does not count an edge length more than once. The following lemma summarizes the required conditions 
for the functions we seek.

\begin{lemma}\label{lem:certificate}
Let $\lambda_1,\lambda_2,\lambda_3\geq0$ and $W=\lambda_1+\lambda_2+\lambda_3$. Suppose there exist four real-valued functions
$q,u_1,u_2,u_3$ defined on $V$ satisfying the following conditions:
\begin{itemize}
\vspace{-1mm}
\item[(i)] $q\geq0$ and $q(r_i)=0$ for each $1\le i\le k$;
\vspace{-1mm}
\item[(ii)] $0\leq u_m\leq\lambda_m$, and $u_m$ takes the value $\lambda_m$ at some receiver and the value $0$ at another receiver when $\lambda_m>0$, 
for each $1\le m \le 3$;
\vspace{-1mm}
\item[(iii)] $|q(v)-q(w)|+\sum_{m=1}^3|u_m(v)-u_m(w)|\leq d(v,w)$ for each pair $v,w\in V$; and 
\vspace{-1mm}
\item[(iv)] $d(s_i,r_i)\leq q(s_i)+W$ for each $1\le i\le k$.
\end{itemize}
Then $\sum_{i=1}^k d(s_i,r_i) \leq\sum_{e\in E}\ell(e)$.
\end{lemma}


{\bf Proof}. By (i), Lemma~\ref{lem:levels}(i) applies to $q$. By (ii), Lemma~\ref{lem:levels}(ii) applies to each $u_m$. Adding all the corresponding 
inequalities and using (iii) and (iv) gives
\begin{align*}
 \sum_{i=1}^k d(s_i,r_i)
 &\leq \sum_{i=1}^k q(s_i)+kW = \sum_{i=1}^k q(s_i)+ \sum_{m=1}^3 k \lambda_m \\
 &\leq\sum_{(v,w)\in A}
 \left([q(v)-q(w)]_++\sum_{m=1}^3|u_m(v)-u_m(w)|\right)\\
 &\leq\sum_{(v,w)\in A}d(v,w) 
 \leq\sum_{e\in E}\ell(e),
\end{align*}
where the last inequality holds, because $d(v,w) \le \ell(e)$ for the edge $e=vw$.  \qed\\

The choice of three functions $u_1,u_2,u_3$, and hence of the three-dimensional transformation introduced below, is motivated by the receiver triple 
chosen in Section 3. The following transformation of three real numbers will enforce the inequality specified in Lemma \ref{lem:certificate}(iii).

For a triple $x=(x_1,x_2,x_3)\in{\mathbb R}^3$, sort the three values with repetitions retained. Let $L(x)$ and $M(x)$ be the smallest and middle values 
in $\{x_1,x_2,x_3\}$, respectively, and let $\mu(x)$ be a subscript $i \in \{1,2,3\}$ with $x_i=L(x)$ (we break ties arbitrarily). Define

\vskip 1.5mm
$\bullet$ $Q(x):=\frac{1}{2}\, [L(x)+M(x)]$,

\vskip 1.5mm
$\bullet$ $H(x):=\frac{1}{2}\, [M(x)-L(x)]$,

\vskip 1.5mm
$\bullet$ $U(x):=(U_1(x), U_2(x), U_3(x))$ with $U_m(x)=H(x)$ if $m=\mu(x)$ and $U_m(x)=0$ otherwise.

\vskip 1.5mm
\noindent Observe that if $L(x)=M(x)$, then the definition of $\mu(x)$ is not unique. Nevertheless, in this case $H(x)=0$ and 
hence $U_m(x)=0$ for all $1\le m \le 3$. So $U(x)$ is well defined.  

\begin{lemma}\label{lem:three}
For any $x,y\in{\mathbb R}^3$, the following inequality holds.

\vspace{-3mm}
\[ |Q(x)-Q(y)|+\sum_{m=1}^3|U_m(x)-U_m(y)|  \leq\max_{1\le m \le 3}|x_m-y_m|.\]
\end{lemma}

\vspace{-1mm}

{\bf Proof}. Let $\delta:=\max_m|x_m-y_m|$. Then 

\vskip 1mm
(1) $y_m-\delta\leq x_m\leq y_m+\delta$ for $1\le m \le 3$.

\vskip 1mm
Let $i:=\mu(x)$ and $j:=\mu(y)$; that is, $L(x)=x_i$ and $L(y)=y_j$. Observe that 

\vskip 1mm
(2) $|L(x)-L(y)|\leq\delta$ and $|M(x)-M(y)|\leq\delta$.

\vskip 1mm
Indeed, by definition and by (1), we have $L(y)-\delta\leq y_m-\delta\leq x_m$ for $1\le m \le 3$. So $L(y)-\delta\leq L(x)$. Furthermore, 
$L(x)\le x_j \leq y_j+\delta=L(y)+\delta$. Combining these two inequalities, we obtain $|L(x)-L(y)|\leq\delta$. Since two or three entries 
of $y$ are at most $M(y)$, the corresponding entries of $x$ are at most $M(y)+\delta$ by (1). Thus $M(x)\leq M(y)+\delta$. By symmetry,
we also have $M(y)\leq M(x)+\delta$. Therefore $|M(x)-M(y)|\leq\delta$. This proves (2).

\vskip 1mm
We proceed by considering two cases.
\vskip 1mm

{\bf Case 1.} $i=j$.

\vskip 1mm
In this case, $\sum_{m=1}^3 |U_m(x)-U_m(y)|=|H(x)-H(y)|$. From the definitions of $Q$ and $H$, we deduce that
\begin{align*}
 |Q(x)-Q(y)|+|H(x)-H(y)|
 &=\frac12\bigl|(L(x)-L(y))+(M(x)-M(y))\bigr|\\
 &\quad+\frac12\bigl|(M(x)-M(y))-(L(x)-L(y))\bigr|\\
 &=\max\{|L(x)-L(y)|,|M(x)-M(y)|\},\\
 &\leq\delta,
\end{align*}
where the second equality follows from the fact that $|a+b|+|a-b|=2\max\{|a|,|b|\}$ for $a, b \in {\mathbb R}$ (which can be seen immediately by squaring both 
sides), and the last inequality follows from (2).

\vskip 1mm
{\bf Case 2.} $i\neq j$.

\vskip 1mm
In this case, $\sum_{m=1}^3|U_m(x)-U_m(y)|=H(x)+H(y)$. Using the fact $|a|=\max\{a,-a\}$ for $a \in {\mathbb R}$ and the definitions of $Q$ and $H$, we obtain
\begin{align*}
 |Q(x)-Q(y)|+H(x)+H(y)
 &=\max\{Q(x)-Q(y),Q(y)-Q(x)\}+H(x)+H(y)\\
 &=\max\{M(x)-L(y),M(y)-L(x)\}.
\end{align*}
Since $i \neq j$, we have $x_j \geq M(x)$. Thus $M(x)-L(y) \leq x_j-y_j \leq |x_j-y_j|\leq \delta$. Similarly, $y_i\geq M(y)$, so $M(y)-L(x)\leq y_i-x_i 
\leq|y_i-x_i|\leq\delta$. Hence $|Q(x)-Q(y)|+H(x)+H(y) \le \delta$. 

\vskip 1mm
Therefore, the desired inequality holds in both cases.  \qed\\

A real-valued function $f$ on $V$ is called \emph{$1$-Lipschitz} with respect to $d$ if 
\[|f(v)-f(w)|\leq d(v,w) \hskip 3mm  \mbox{for any} \hskip 2mm v, w\in V.\] 

\noindent Suppose we have got three $1$-Lipschitz functions $F_1,F_2,F_3$. Define $F(v):=(F_1(v),F_2(v),F_3(v))$, $q(v):=Q(F(v))$, 
and $u_m(v):=U_m(F(v))$ for $v\in V$ and $1\le m \le 3$. By Lemma~\ref{lem:three}, for any $v,w\in V$, we have
\begin{align*}
 |q(v)-q(w)|+\sum_{m=1}^3|u_m(v)-u_m(w)|
 &\leq\max_{1\leq m\leq3}|F_m(v)-F_m(w)|\\
 &\leq d(v,w),
\end{align*}
which establishes condition~(iii) of Lemma~\ref{lem:certificate}. To ensure that $q(r_i)=0$ for each $1\le i \le k$, we shall prescribe $F(r_i)=g(i)$, 
where $g(i)=(g_1(i),g_2(i),g_3(i))$ is a permutation of $(-h,h,h)$ for some $h\geq0$ that may depend on $i$. Such a triple satisfies $Q(g(i))=0$, 
and its $U$-coordinates are a permutation of $(h,0,0)$. We construct these triples in Section 3 and extend the prescribed values to the functions 
$F_1,F_2,F_3$ on $V$ in Section 4.

\section{Prescribed Values at Receivers}

Recall that $g_m(i)$ is the value prescribed for $F_m(r_i)$ for $1\le m \le 3$. To find these values properly, let
us first determine the constants $W, \lambda_1, \lambda_2, \lambda_3$ as described in Lemma~\ref{lem:certificate}.

According to Lemma~\ref{lem:certificate}(iv), $d(s_i,r_i)\leq q(s_i)+W$ for some function $q: V\to{\mathbb R}_+$ and constant $W \ge 0$.
We shall actually establish that $d(v,r_i)\leq q(v)+W$ for all $v \in V$.

For any receiver labels $1\le i,j \le k$, put $\rho_{ij}:=d(r_i,r_j)$. For any three receiver labels $i,j,h$, define their {\em perimeter} by 
$P(i,j,h):=\rho_{ij}+\rho_{ih}+\rho_{jh}$. The triangle inequality gives $d(v,r_i)\leq d(v,r_j)+\rho_{ij}$ and $d(v,r_i)\leq d(v,r_h)+\rho_{ih}$. 
Adding these inequalities and dividing the sum by two, we obtain
\begin{equation*}\label{eq:perimeter-comparison}
\begin{aligned}
 d(v,r_i)
 &\leq
 \frac{d(v,r_j)+d(v,r_h)+\rho_{ij}+\rho_{ih}}2\\
 &=
 \frac{d(v,r_j)+d(v,r_h)-\rho_{jh}}2
 +\frac{P(i,j,h)}2.
\end{aligned}
\end{equation*}
Clearly, the first term is nonnegative by the triangle inequality. As the second term depends only on three receivers, it is
natural to upper-bound it by 

\vskip 1mm
{\bf (3.1)} $W=\frac12\max_{1 \le i,\, j,\, h \le k} P(i,j,h).$
\vskip 1mm

Observe that if $k\geq3$, then the maximum value is achieved at three distinct labels. Indeed, a triple with only two distinct labels $i,j$ 
has perimeter $2\rho_{ij}$. For any third label $h$, the triangle inequality gives $2\rho_{ij}\leq\rho_{ij}+\rho_{ih}+\rho_{jh}=P(i,j,h)$.
Moreover, a triple with all labels equal has perimeter $0$. Thus, renaming the labels if necessary, we may assume (see (3.1)) that

\vskip 1mm
{\bf (3.2)} $W= \frac12 P(1,2,3)= \frac12\max_{1 \le i,\, j,\, h \le k} P(i,j,h)$.
\vskip 1mm

We postpone until the next section the definition of the function $q$, which satisfies $d(v, r_i)-W \le q(v)$ for all $v\in V$ and $1\le i \le k$.

Let us now choose $\lambda_1,\lambda_2,\lambda_3$ so that all the conditions specified in Lemma~\ref{lem:certificate} are satisfied. To ensure condition (ii),
we impose that $u_m(r_m)=\lambda_m$ and $u_m(r_n)=0$ for $1\le m\ne n \le 3$. Since $q$ must vanish at all receivers, applying condition (iii) 
to $r_m$ and $r_n$ yields $\lambda_m+\lambda_n\leq\rho_{mn}$ for $1\le m \ne n \le 3$. We impose equality in all three inequalities; that is,

\vskip 1mm
{\bf (3.3)} $\lambda_m+\lambda_n=\rho_{mn}$ for $1\le m \ne n \le 3$.
\vskip 1mm

Solving this linear system gives

\vskip 1mm
{\bf (3.4)} $\lambda_m= \frac{1}{2} (\rho_{mn}+\rho_{mt}-\rho_{nt})$ for $\{m,n,t\}=\{1,2,3\}$. Obviously, 
$\lambda_m\geq 0$ and $\lambda_1+\lambda_2+\lambda_3=W$. 

\vskip 1mm

Recall the notation $L(x)$, $M(x)$, $\mu(x)$, $Q(x)$, $H(x)$, and $U(x)$ for $x \in{\mathbb R}^3$ introduced right above
Lemma \ref{lem:three}. Since the distances to a fixed receiver change by at most $\rho_{ij}$ when the receiver label is changed 
from $i$ to $j$, they provide a natural starting point for prescribing the values at receivers, as we require $1$-Lipschitz functions. 
Thus, for each receiver label $1\le i \le k$ and position $1\le m \le 3$, we define

\vskip 1mm
$\bullet$ $B(i)=\bigl(B_1(i),B_2(i),B_3(i)\bigr)$, where $B_m(i)=\rho_{im}-\lambda_m$; 
\vskip 1mm
$\bullet$ $g(i)=\bigl(g_1(i),g_2(i),g_3(i)\bigr)$, where $g_m(i)=-H(B(i))$ if $m=\mu(B(i))$ and $g_m(i)=H(B(i))$ otherwise. 
\vskip 1mm

Observe that the triple $B(i)$ consists of three adjusted distances. Since $Q(B(i))$ is not necessarily $0$, $B(i)$ cannot serve as 
$g(i)$; its entries must therefore be further adjusted to obtain $g(i)$. Moreover, if $L(B(i))=M(B(i))$, then $H(B(i))=0$, 
so $g(i)$ is well defined.

The following lemma shows that the values $Q(g(i))$ and $U_m(g(i))$ satisfy the conditions required of $q$ and $u_m$ at the receivers.

\begin{lemma}\label{zang}
For any receiver labels $1\le i,j\le k$ and positions $1\le m,n\le 3$, the following statements hold.
\begin{itemize}
\vspace{-1mm}
\item[(i)] $Q(g(i))=0$, $H(g(i))=H(B(i))$, and $U(g(i))=U(B(i))$;
\vspace{-1mm}
\item[(ii)] $0\leq U_m(g(i))\leq\lambda_m$;
\vspace{-1mm}
\item[(iii)] $U_m(g(m))=\lambda_m$ and $U_n(g(m))=0$ when $n\neq m$; and
\vspace{-1mm}
\item[(iv)] $\sum_{m=1}^3|U_m(g(i))-U_m(g(j))|\leq\rho_{ij}$.
\end{itemize}
\end{lemma}

{\bf Proof}. (i) By definition, $g_m(i)=-H(B(i))$ if $m=\mu(B(i))$ and $g_m(i)=H(B(i))$ otherwise. So $L(g(i))=-H(B(i))$ and $M(g(i))=H(B(i))$. 
Hence $Q(g(i))=\frac{1}{2}[L(g(i))+M(g(i))]=0$ and $H(g(i))=\frac{1}{2}[M(g(i))-L(g(i))]=H(B(i))$. Moreover, $U_m(g(i))=H(g(i))=H(B(i))$ if $m=\mu(B(i))$
and $U_m(g(i))=0$ otherwise. Therefore, $U(g(i))=U(B(i))$.

\vskip 1mm
(ii) Recall that $U_m(g(i))=H(B(i))$ if $m=\mu(B(i))$ and $U_m(g(i))=0$ otherwise. So we only need to consider the case when $m=\mu(B(i))$.
Since $H(B(i))=\frac{1}{2}[M(B(i))-L(B(i)]$, we have $H(B(i))= \frac{1}{2} \min_{n\neq m} \bigl[B_n(i)-B_m(i)\bigr]$. By the triangle inequality, 
$B_n(i)-B_m(i)= (\rho_{in}-\lambda_n)- (\rho_{im}-\lambda_m) \leq\rho_{nm}-\lambda_n+\lambda_m  =2\lambda_m$ by (3.3). So $U_m(g(i))\leq\lambda_m$.

\vskip 1mm
(iii) By definition and (3.3), we have $B_m(m)=-\lambda_m$ and $B_n(m)=\rho_{mn}-\lambda_n=\lambda_m$ for $n\neq m$.
Thus $H(B(m))=\lambda_m$ and $\mu(B(m))=m$. By (i), we also have $H(g(m))=\lambda_m$. Moreover, by definition, $g_m(m)=-\lambda_m$ and $g_n(m)=\lambda_m$ 
for $n\neq m$. Hence $U_m(g(m))=\lambda_m$ and $U_n(g(m))=0$ when $n\neq m$.

\vskip 1mm
(iv) For each $1\le m \le 3$, the triangle inequality gives $|B_m(i)-B_m(j)|=|\rho_{im}-\rho_{jm}|\leq\rho_{ij}$.
Applying Lemma~\ref{lem:three} to $B(i),B(j)$ and dropping the nonnegative term $|Q(B(i))-Q(B(j))|$ yields
$\sum_{m=1}^3|U_m(B(i))-U_m(B(j))|  \leq\max_{1\leq m\leq3}|B_m(i)-B_m(j)|\leq\rho_{ij}$. From (i) we deduce that
$\sum_{m=1}^3|U_m(g(i))-U_m(g(j))|\leq\rho_{ij}$. \qed\\

Let us show that, for each $1\le m \le 3$, the function $g_m$ is 1-Lipschitz with respect to the receiver labels, and 
establish additional properties of each $g_m$ that allow us to extend the 1-Lipschitz property to all vertices.

\begin{lemma}\label{lem:g}
For any receiver labels $1\le i,j\le k$ and positions $1\le m,n\le 3$, the following statements hold.
\begin{itemize}
\vspace{-1mm}
\item[(i)] $g_m(i)=-U_m(g(i))+\sum_{n \neq m} U_n (g(i))$;
\vspace{-1mm}
\item[(ii)] $|g_m(i)-g_m(j)|\leq\rho_{ij}$;
\vspace{-1mm}
\item[(iii)] $g_m(i)+g_n(j)\geq-\rho_{ij}$ when $m\neq n$; and
\vspace{-1mm}
\item[(iv)] $g_n(i)-g_m(i)\leq2\lambda_m$ when $m\neq n$.
\end{itemize}
\end{lemma}

{\bf Proof}. (i) By definition, $g_m(i)=-H(B(i))$ if $m=\mu(B(i))$ and $g_m(i)=H(B(i))$ otherwise. By Lemma \ref{zang}(i), 
$U_m(g(i))=H(B(i))$ if $m=\mu(B(i))$ and $U_m(g(i))=0$ otherwise. Depending on whether $m=\mu(B(i))$, we consider two cases. It is straightforward 
to verify that the desired equality holds in both cases.    

\vskip 1mm
(ii) By (i) and by Lemma~\ref{zang}(iv), we obtain $|g_m(i)-g_m(j)|\leq\sum_{n=1}^3|U_n(g(i))-U_n(g(j))|\leq\rho_{ij}$. 

\vskip 1mm
(iii) Since $1\le m \ne n \le 3$, we have $g_m(i)+g_n(i)\geq 0$. So $g_m(i)+g_n(j)=g_m(i)+g_n(i)-\bigl(g_n(i)-g_n(j)\bigr)\geq -|g_n(i)-g_n(j)|\geq
-\rho_{ij}$ by (ii).

\vskip 1mm
(iv) From (i) and Lemma~\ref{zang}(ii), we deduce that $g_n(i)-g_m(i)=2U_m(g(i))-2U_n(g(i)) \leq 2U_m(g(i)) \leq2 \lambda_m$.  \qed

\section{Proof of the Main Result}

For each position $1\le m \le 3$, Lemma \ref{lem:g}(ii) asserts that the function $g_m$ is 1-Lipschitz with respect to the receiver labels. 
Let us now extend the prescribed receiver values and the 1-Lipschitz property to all vertices. Motivated by McShane's extension theorem
\cite{M}, we define 
\[F_m(v):=\min_{1\le i \le k} \, \bigl(d(v,r_i)+g_m(i)\bigr)\]
\noindent for $v\in V$ and $1\le m \le 3$. 

\begin{lemma}\label{lem:F}
For each $1\le m \le 3$, the function $F_m$ satisfies $F_m(r_i)=g_m(i)$ for $1\le i \le k$, and $F_m$ is $1$-Lipschitz
on $V$.
\end{lemma}

\vspace{-1mm}

{\bf Proof}. By definition, $F_m(r_i)=\min_{1\le j \le k} \, \bigl(d(r_i,r_j)+g_m(j)\bigr)= \min_{1\le j \le k} \, \bigl(\rho_{ij}+g_m(j)\bigr)$.
By Lemma~\ref{lem:g}(ii),  $\rho_{ij}+g_m(j)\geq g_m(i)=\rho_{ii}+g_m(i)$. So the minimum value is achieved when $j=i$; that is,  
$F_m(r_i)=g_m(i)$.

Next, for any $v,w\in V$ and $1\le j\le k$, the triangle inequality gives $d(v,r_j)+g_m(j)\leq d(v,w)+d(w,r_j)+g_m(j)$.
Taking the minimum over $j$ yields $F_m(v)\leq d(v,w)+F_m(w)$. By symmetry, the inequality $F_m(w)\leq d(v,w)+F_m(v)$ also holds. 
So $|F_m(v)-F_m(w)|\leq d(v,w)$. Therefore $F_m$ is $1$-Lipschitz on $V$.  \qed\\

Apply the transformations introduced right above Lemma \ref{lem:three} to $F(v)=(F_1(v),F_2(v), F_3(v))$, and set
\[q(v):=Q(F(v)) \text{  and  } u_m(v):=U_m(F(v)) \text{  for  } v\in V \text{  and  } 1\le m \le 3.\]

Let us show that they satisfy all the requirements of Lemma~\ref{lem:certificate}.

\begin{lemma}\label{lem:vertex}
For the functions $q, u_1, u_2, u_3$, the following statements hold.
\begin{itemize}
\vspace{-1mm}
\item[(i)] $q(v)\geq0$ for each $v\in V$, and $q(r_i)=0$ for each $1\le i \le k$;
\vspace{-1mm} 
\item[(ii)]  $u_m(r_m)=\lambda_m$ and $u_m(r_n)=0$ for each $1\le m\ne n \le 3$;
\vspace{-1mm}
\item[(iii)] $0\leq u_m(v)\leq\lambda_m$ for each $v\in V$ and $1\le m \le 3$; and
\vspace{-1mm}
\item[(iv)] $|q(v)-q(w)|+\sum_{m=1}^3|u_m(v)-u_m(w)|\leq d(v,w)$ for each pair $v,w\in V$.
\end{itemize}
\end{lemma}

{\bf Proof}. (i) By definitions, $q(v)=Q(F(v))= \frac{1}{2} [L(F(v))+M(F(v))]= \frac{1}{2} \min_{1\leq m<n\leq 3} \, (F_m(v)+F_n(v))$. 
For fixed $m <n$, we have $F_m(v)+F_n(v) = \min_{1\le i \le k} \, \bigl(d(v,r_i)+g_m(i)\bigr)
+ \min_{1\le j \le k}\, \bigl(d(v,r_j)+g_n(j)\bigr)= \min_{1\le i, \, j \le k} \, \bigl(d(v,r_i)+d(v,r_j)+g_m(i)+g_n(j)\bigr)$.
By Lemma~\ref{lem:g}(iii) and the triangle inequality, $d(v,r_i)+d(v,r_j)+g_m(i)+g_n(j) \geq d(v,r_i)+d(v,r_j)-\rho_{ij}\geq 0$.
So  $q(v)\geq0$ for $v\in V$. By Lemma~\ref{lem:F}, $F(r_i)=g(i)$, which together with Lemma~\ref{zang}(i) implies that 
$q(r_i)=Q(F(r_i))=Q(g(i))=0$.

\vskip 1mm 

(ii) Since $F(r_i)=g(i)$, we have $u_m(r_i)=U_m(F(r_i))=U_m(g(i))$. From Lemma~\ref{zang}(iii) we see that $u_m(r_m)=\lambda_m$ and $u_m(r_n)=0$ for 
$1\le m\ne n \le 3$. 

\vskip 1mm

(iii) By definition, $u_m(v)=U_m(F(v))$ for $v\in V$. Clearly, $u_m(v) \geq 0$. If $u_m(v)=0$, then $u_m(v) \le \lambda_m$ by
(3.4). So we may assume that $u_m(v)>0$. Hence, from the definition of $U(x)$, we see that $F_m(v)$ is the unique smallest entry of 
$F(v)$. Let $n$ be the subscript with $1\le n \le 3$ such that $F_n(v)=M(F(v))$. Then $n \ne m$. Recall the definition of $F_m(v)$.
Let $t$ be a subscript with $1 \le t \le k$ such that $F_m(v)=d(v,r_t)+g_m(t)$. Then $F_n(v)\leq d(v,r_t)+g_n(t)\leq d(v,r_t)+g_m(t)+
2\lambda_m=F_m(v)+2\lambda_m$, where the second inequality follows from Lemma~\ref{lem:g}(iv). This bound implies that $u_m(v)=
\frac{1}{2} [F_n(v)-F_m(v)] \leq\lambda_m$.  

\vskip 1mm
(iv) By Lemma~\ref{lem:three} and the definitions of $q,u$, we have $|q(v)-q(w)|+\sum_{m=1}^3|u_m(v)-u_m(w)|
= |Q(F(v))-Q(F(w))|+\sum_{m=1}^3|U_m(F(v))-U_m(F(w))| \leq\max_{1\leq m\leq3} |F_m(v)-F_m(w)| \leq d(v,w)$, where the last 
inequality follows from the $1$-Lipschitz property in Lemma~\ref{lem:F}.  \qed\\

The following lemma establishes that an even stronger form of condition (iv) in Lemma~\ref{lem:certificate} holds.

\begin{lemma}\label{lem:source}
The inequality $d(v,r_i)\leq q(v)+W$ holds for all $v\in V$ and $1\le i \le k$. 
\end{lemma}

{\bf Proof}. By definitions, $q(v)=Q(F(v))= \frac{1}{2} [L(F(v))+M(F(v))]= \frac{1}{2} \min_{1\leq m<n\leq 3} \, (F_m(v)+F_n(v))$. 
For fixed $m <n$, we have $F_m(v)+F_n(v) = \min_{1\le j \le k} \, \bigl(d(v,r_j)+g_m(j)\bigr) + \min_{1\le h \le k} \, \bigl(d(v,r_h) \\+g_n(h)\bigr)= 
\min_{1\le j, \, h \le k} \, \bigl(d(v,r_j)+d(v,r_h)+g_m(j)+g_n(h)\bigr)$. By Lemma~\ref{lem:g}(iii) and the triangle inequality, 
$d(v,r_j)+d(v,r_h)+g_m(j)+g_n(h) \geq d(v,r_j)+d(v,r_h)-\rho_{jh}= (d(v,r_j)+\rho_{ij})+(d(v,r_h)+\rho_{ih})-P(i,j,h) 
\geq 2d(v,r_i)-P(i,j,h)$.  By (3.2), we thus obtain

\vskip 1mm
(1) $d(v,r_j)+d(v,r_h)+g_m(j)+g_n(h) \geq 2d(v,r_i)-2W.$
\vskip 1mm

Since inequality (1) holds for all $1 \le j, h \le k$, it remains valid after taking the minimum over all of them. Consequently,

\vskip 1mm
(2) $F_m(v)+F_n(v)\geq 2d(v,r_i)-2W$.
\vskip 1mm

Since inequality (2) holds for all $m< n$, we deduce that $q(v)=\frac{1}{2} \min_{1\le m< n \le 3} \, \bigl(F_m(v)+F_n(v)\bigr)
\geq d(v,r_i)-W$. Therefore $d(v,r_i)\leq q(v)+W$ for all $v\in V$ and $1\le i \le k$. \qed\\

We are now ready to establish the main result.

\vskip 1mm
{\bf Proof of Theorem~\ref{cor:length}.} By (3.4), we have $\lambda_1,\lambda_2,\lambda_3\geq0$ and $W=\lambda_1+\lambda_2+\lambda_3$. By
Lemmas \ref{lem:vertex} and \ref{lem:source}, the functions $q,u_1,u_2,u_3$ constructed above satisfy all the conditions of 
Lemma~\ref{lem:certificate}.  Therefore $\sum_{i=1}^k d(s_i,r_i)\leq\sum_{e\in E}\ell(e)$ by Lemma~\ref{lem:certificate}. \qed

\end{document}